%% file: main.tex
\def\draft{0}

\documentclass[a4paper,UKenglish,cleveref, autoref, thm-restate, cleveref]{lipics-v2021}

\pdfoutput=1 
\hideLIPIcs  

\nolinenumbers 

\usepackage{amsmath}
\usepackage{amsthm}
\usepackage{amsfonts}
\usepackage{amssymb}
\usepackage{mdframed}
\usepackage{chemmacros}
\usepackage{xcolor}
\usepackage{soul}
\usepackage{svg} 

\usepackage{dsfont}

\usepackage{upgreek}
\usepackage{chemmacros}

\chemsetup{
  formula = mhchem,
  reactions/tag-open=(R,
  reactions/tag-close=),
}

\usepackage[textsize=tiny,textwidth=0.9\marginparwidth]{todonotes} 

\usepackage{booktabs}
\usepackage{apxproof}

\usepackage{etoolbox}
\BeforeBeginEnvironment{thebibliography}{\par\begingroup\nolinenumbers}
\AfterEndEnvironment{thebibliography}{\endgroup}

\input{macro}

\usepackage{tikz}
\usetikzlibrary{shapes.geometric,shapes.misc,arrows.meta}

\title{Amplification at Equilibrium: Structural and Thermodynamic Limitations, and Implementation} 

\titlerunning{Amplification at Equilibrium: Limitations and Implementation} 

\author{Hamidreza Akef}{Department of Electrical and Computer Engineering, The University of Texas at Austin, TX, USA}{hakef@utexas.edu}
{https://orcid.org/0009-0001-9253-9737}
{}
\author{Chia-Yu Sung}{Department of Molecular Biosciences, The University of Texas at Austin, TX, USA}
{chiayusung@utexas.edu}
{https://orcid.org/0009-0001-7547-7159}
{}

\author{Aneesh Vanguri}{Department of Molecular Biosciences, The University of Texas at Austin, TX, USA}
{aneeshv@utexas.edu}
{}
{}

\author{David Soloveichik}{Department of Electrical and Computer Engineering, The University of Texas at Austin, TX, USA \and \url{https://www.solo-group.link/}}{david.soloveichik@utexas.edu}{https://orcid.org/0000-0002-2585-4120}{}

\authorrunning{H. Akef, C.Y. Sung, A. Vanguri, and D. Soloveichik} 

\Copyright{Hamidreza Akef, Chia-Yu Sung, Aneesh Vanguri, and David Soloveichik} 



\category{} 

\funding{Schmidt Sciences Polymath Award to D.S., DOE grant DE-SC0024248, NSF SemiSynBio III: GOALI grant 2227578}

\acknowledgements{We thank Antoni Szeglowski for performing experiments and Marko Vasi\'c for sequence design on an earlier variation of the trimeric equilibrium amplifier.}

\EventEditors{Dominic Scalise and Robert Schweller}
\EventNoEds{2}
\EventLongTitle{32nd International Conference on DNA Computing and Molecular Programming (DNA 32)}
\EventShortTitle{DNA 32}
\EventAcronym{DNA}
\EventYear{32}
\EventDate{August 3--7, 2026}
\EventLocation{Fayetteville, Arkansas, US}
\EventLogo{}
\SeriesVolume{387}
\ArticleNo{}

\begin{document}

\maketitle

\begin{abstract}
Amplifying weak molecular signals is essential in both natural and engineered biochemical systems. While most amplification schemes operate out of equilibrium, relying on kinetic barriers and fuel-driven cascades, it is also possible to amplify at thermodynamic equilibrium by shifting the energy landscape upon addition of an analyte. Equilibrium amplification is appealing because, in principle, the system can remain indefinitely in the untriggered state. In this work, we establish fundamental structural and thermodynamic limits on equilibrium-based amplification. We first prove that dimerization networks---systems restricted to complexes of at most two monomers---are inherently incapable of equilibrium amplification. 
This no-go theorem explains the absence of amplification in prior undercomplementary ``strand commutation'' designs. 
We then show that allowing trimeric complexes breaks this barrier. 
We propose an isometric trimer-based equilibrium amplifier whose output preserves the size of the input, enabling modular composition, and validate it experimentally, achieving an amplification factor close to the expected $2\times$. 
Finally, we derive universal thermodynamic bounds applicable regardless of complex size: 
the maximum amplification factor scales linearly with the free energy of interaction between the analyte and the amplifier components. 
For nucleic acid systems, this implies that the analyte length must grow linearly with the desired amplification factor, and that composing modular amplifiers yields diminishing returns for a fixed analyte. 
Together, these results delineate the structural and energetic boundaries of equilibrium amplification and rigorously justify the necessity of out-of-equilibrium approaches for achieving high gain.
\end{abstract}

\section{Introduction}

In both natural and engineered systems, amplifying weak molecular signals is crucial for reliable function. For example, signal transduction cascades employ multiple mechanisms---including second messengers, protein phosphorylation, and enzyme cascades---to convert low-level extracellular inputs into robust intracellular responses~\cite{Su2024}. Similarly, regulatory RNAs act as ultrasensitive switches to tightly govern gene expression~\cite{Levine2007}. This broad capacity for signal amplification is essential for maintaining cellular specificity and responsiveness across diverse biological contexts, from hormone signaling to receptor-mediated pathways that drive cellular differentiation and metabolism~\cite{Chen2024,Liu2026}.

Parallel to these natural architectures, amplification in molecular diagnostics and synthetic biology enables the detection of scarce biomarkers and the construction of complex biochemical circuits. Enzyme-free DNA circuits have demonstrated highly sophisticated functions: Seelig~\emph{et al.} implemented modular logic gates with signal restoration and amplification~\cite{seelig2006science}, and Qian and Winfree scaled strand-displacement cascades into multi-layer logic devices~\cite{Lulu2011Science}. Likewise, dynamic amplification schemes such as the hybridization chain reaction and catalytic hairpin assembly rely on cascading strand displacements to robustly boost signal levels~\cite{Li2011}.

Chemical amplifiers typically rely on bridging an energy barrier in the presence of the analyte and thus necessarily operate out of equilibrium.
For example, most DNA‐based amplifiers rely on fuel-driven kinetic cascades. 
In the hybridization chain reaction, two DNA hairpins remain kinetically trapped until an initiator strand opens one, setting off a chain of strand-displacement reactions~\cite{Dirks-Pierce2004}.
Likewise, catalytic hairpin assembly uses a series of metastable hairpin fuels: analyte binding to one hairpin exposes a toehold that catalytically opens the next, releasing the analyte to repeat the cycle and yielding hundreds-fold signal amplification~\cite{Jiang2013}.

Nonetheless, it is possible to perform signal amplification at thermodynamic equilibrium.
In such systems, prior to the addition of the analyte, the thermodynamic equilibrium has a relatively small amount of output (detection species).
The addition of the analyte shifts the energy landscape, and subsequent relaxation to the new equilibrium increases the output concentration by more than the amount of analyte added.
An important potential advantage of equilibrium amplifiers is that they do not rely on a kinetic barrier, and thus can indefinitely remain in the untriggered state.
Equilibrium amplifiers also necessarily avoid complex enzymes typical of natural signaling cascades
since the only role an enzyme can play is to modify a kinetic barrier.
Equilibrium amplifiers also support dynamic signal processing: changing the input (e.g., swapping input strands) causes the system to re-equilibrate accordingly~\cite{sterin2025thermodynamically}, whereas kinetically controlled systems are generally single-use, requiring their metastable complexes to be externally regenerated for each new input.
A family of equilibrium amplifiers was recently proposed in the Thermodynamic Binding Network (TBN) model~\cite{Doty2017,petrack2023}.

Equilibrium-based chemical computation recently received a significant boost with the work of Nikitin on ``strand commutation'' networks, who showed that diverse computations can be performed by weak hybridization of DNA strands~\cite{nikitin2023noncomplementary}.
Specifically, the inputs are DNA strands and the computation arises via the competition of multiple possible binding partners between pairs of strands.
The thermodynamic equilibrium of the input together with the rest of the system then contains the desired amount of output strands.
Undercomplementarity appears to allow significantly more compact constructions. 
For example, Nikitin contrasts prior work on the square root circuit using 130 strands~\cite{Lulu2011Science} with his undercomplementary implementation using only nine strands (plus four inputs).
Nonetheless, the computations Nikitin demonstrated suffer from significant signal attenuation, especially in the more complex constructions, which exhibit output concentrations many orders of magnitude lower than the input concentration (e.g.~analog computation). 
In none of the proposed systems in~\cite{nikitin2023noncomplementary} is the amount of output greater than the amount of input,
suggesting that some aspect of Nikitin's paradigm prevents amplification.

In the current paper, we study amplification at thermodynamic equilibrium, using both theoretical arguments and experimental realization. 
Our theoretical results are general and apply to DNA-based systems as well as other substrates.
In \cref{sec:dimerization}, we consider the natural classification of equilibrium systems according to maximum complex size, and show that dimerization networks (maximum complex size two) cannot amplify.
This result explains the lack of amplification in Nikitin's work despite the otherwise extensive functional diversity of dimerization networks.

In contrast to dimerization networks, allowing trimers permits equilibrium amplification. 
In \cref{sec:trimerization}, we develop a novel undercomplementary trimeric $2\times$ DNA amplifier,
and demonstrate that it is able to achieve close to ideal amplification in wet-lab experiments.
In contrast to the previously proposed equilibrium amplifier in the TBN model~\cite{petrack2023}, our undercomplementary implementation is substantially more compact in terms of the number of strands overall, the size of the largest complexes, and the number of logical domains. 
Similar to the TBN amplifier, our amplifier is ``composable'':
the output strand has the same form as the input strand, and could act as an input to a downstream amplifier of the same design.
Composability suggests a method to achieve overall increased amplification. 

In the final part of this paper (\cref{sec:thermo_bounds}), rather than focusing on structural constraints that limit equilibrium amplification, we derive bounds on the amplification factor based on the free energy of interaction between the analyte and the rest of the system. 
In the context of nucleic acid systems, our bound implies that the length of the detected analyte must scale linearly with the desired equilibrium amplification factor.
These results formalize the intuition that a larger difference between the triggered and untriggered configurations---necessary for a greater degree of amplification---necessarily incurs an energetic cost,
and we show that this cost scales linearly with the amplification factor.
Since the analyte is the only strand added to the system, and all other strands remain unchanged, with only the complexes they form being altered, the requisite energy must originate from the direct interaction of the analyte with the other amplifier strands.
For modular equilibrium amplifiers such as our trimeric system and the TBN amplifier~\cite{petrack2023}, the impossibility result implies that the composition of more amplifier modules yields diminishing returns for a fixed analyte.

Thermodynamic limits on biochemical response have been studied most extensively for nonequilibrium sensitivity and ultrasensitivity. For example, Qian showed that the free energy from ATP/GTP hydrolysis controls the achievable ultrasensitivity in phosphorylation-dephosphorylation switches~\cite{Qian2003ThermoKinetic}, and subsequent work derived universal thermodynamic response bounds for nonequilibrium steady states and tight energy bounds for covalent-modification switches~\cite{OwenGingrichHorowitz2020,OwenTallaBiddleGunawardena2023}. A related line of work bounds Hill-type sensitivity by structural quantities such as the number of binding sites or the support of the perturbation~\cite{OwenHorowitz2021SizeLimits}. These results are complementary to the question addressed here, but do not directly address equilibrium amplification as we define it.

Our results establish fundamental principles of equilibrium-based amplification based on structural network topology (i.e., dimerization versus trimerization), as well as thermodynamic constraints.
Our trimeric equilibrium amplifier could potentially be the missing piece for Nikitin-style undercomplementary schemes, allowing computation with less signal loss than is possible in dimerization-only equilibrium networks.
More generally, however, our thermodynamic impossibility results rigorously justify out-of-equilibrium amplifiers as unavoidable for greater amplification in the next-generation ultrasensitive molecular sensors and signal processing architectures.

\section{Structural Limitations: Dimerization Networks Cannot Amplify}
\label{sec:dimerization}

In this section, we analyze the amplification capacity of a chemical reaction network limited to dimerization reactions (i.e., complexes of at most two monomers). 
By amplification, we mean that an increase in the input monomer leads to a larger-than-proportional increase in the concentration of some monomer species or dimer species in the network. 
We rigorously demonstrate that dimerization networks are structurally incapable of signal amplification, regardless of the thermodynamic parameters. 
Specifically, we prove that the sensitivity of any monomer or dimer species in the network to an input signal is bounded by unity.

While previous work~\cite{Maslov2007} used numerical simulations to argue that dimerization networks cannot amplify, rigorous systems-level understanding of these architectures remains elusive~\cite{Parres-Gold2025}.
This gap persists even as researchers increasingly push the boundaries of these architectures, such as utilizing competitive dimerization networks as analog physical computers trained via directed evolution~\cite{Tkachenko2025}.

\subsection{Model Definition and Equilibrium}

Consider a system comprising $n$ distinct monomeric species $\{X_1, \dots, X_n\}$. These monomers can interact to form dimeric species $D_{ij}$ according to the reversible reactions:
\begin{equation*}
    X_i + X_j \rightleftharpoons D_{ij} \quad \text{for } 1 \le i \le j \le n,
\end{equation*}
with the convention $D_{ji}=D_{ij}$. The system is assumed to be in thermodynamic equilibrium. The equilibrium concentrations, denoted by lowercase variables $x_i$ and $d_{ij}$, satisfy the mass-action law:
\begin{equation}\label{eq:eq_conc}
    d_{ij} = K_{ij} \, x_i \, x_j,
\end{equation}
where $K_{ij} \ge 0$ is the equilibrium association constant. The case $K_{ij}=0$ represents an absent dimer, while $K_{ij}>0$ corresponds to an allowed dimerization reaction with finite formation free energy.

We define the total concentration of monomer $X_i$, denoted by $C_i$, as the sum of the free monomer concentration and the concentration of $X_i$ sequestered in dimers. According to the conservation of mass:
\begin{equation}\label{eq:mass_conservation}
    C_i = x_i + \sum_{j \neq i} d_{ij} + 2 d_{ii}.
\end{equation}

\subsection{Sensitivity Analysis}

We examine the system's equilibrium response to adding total concentration $\sigma$ of the input monomer $X_1$. 
Specifically, we aim to derive an equation for the relative sensitivity of each monomer species, defined as the derivative of its equilibrium concentration with respect to the added input amount, normalized by its concentration.
Formally, we say the \emph{sensitivity} of species $X_i$ is $x_i' = \frac{d x_i}{d \sigma}$, and its \emph{relative sensitivity} is $\frac{x_i'}{x_i}$.

For added input amount $\sigma$, the total concentration of monomer $X_i$ is $C_i = C_i^{\text{init}} + \sigma \delta_{1i}$, where $\delta_{1i}$ is the Kronecker delta.
Substituting \cref{eq:eq_conc} into \cref{eq:mass_conservation}, we obtain the closed algebraic mass-balance equations:
\begin{equation*}
    x_i + \sum_{j \neq i} K_{ij} \, x_i \, x_j + 2 K_{ii} \, x_i^2 = C_i^{\text{init}} + \sigma \delta_{1i}.
\end{equation*}
To determine the sensitivity of the system, we differentiate with respect to the signal $\sigma$ to obtain:
\begin{equation*}
    x_i' + \sum_{j \neq i} K_{ij} ( x_i' \, x_j + x_i \, x_j') + 4 K_{ii} \, x_i \, x_i' = \delta_{1i}.
\end{equation*}
Using the equilibrium relationship $K_{ij} x_i x_j = d_{ij}$, we can rewrite $K_{ij} x_j = d_{ij}/x_i$ and $K_{ij} x_i = d_{ij}/x_j$ to relate the relative sensitivities to concentrations of dimeric species:
\begin{equation}\label{eq:mass-cons-der}
    \frac{x_i'}{x_i} \bigg( x_i + \sum_{j \neq i} d_{ij} + 4 d_{ii} \bigg) + \sum_{j \neq i} d_{ij} \frac{x_j'}{x_j} = \delta_{1i}.
\end{equation}

\subsection{The Upper Bound on Amplification}

The following theorem shows that no monomer species in the dimerization network can exhibit a relative sensitivity magnitude greater than that of the input monomer, and that the input monomer species itself cannot be amplified beyond the added signal.
The theorem is proven in \cref{proof:input-domination}.

\begin{theorem}[The Input Domination Theorem]
\label{thm:input-domination}
    Let $\sigma$ represent the total added concentration of input monomer $X_1$.
    In a dimerization network at equilibrium, the sensitivity of the free concentration $x_1$ to changes in the total input is bounded by 1, i.e., $0 < \frac{dx_1}{d\sigma} \le 1$. Furthermore, for any monomer species $X_j$, the relative sensitivity satisfies $|\frac{x_j'}{x_j}| \le |\frac{x_1'}{x_1}|$.
\end{theorem}

\begin{appendixproof}[Proof of \cref{thm:input-domination}]
\label{proof:input-domination}
    Let $k$ be the index of the monomer species with the maximum absolute relative sensitivity. That is
    $\left| \frac{x_k'}{x_k} \right| = \max_j \left| \frac{x_j'}{x_j} \right|$.

    We analyze \cref{eq:mass-cons-der} for the index $k$ with maximum relative sensitivity:
    \begin{equation}\label{eq:proof_step1}
        \frac{x_k'}{x_k} \bigg( x_k + \sum_{j \neq k} d_{kj} + 4 d_{kk} \bigg) + \sum_{j \neq k} d_{kj} \frac{x_j'}{x_j} = \delta_{1k}.
    \end{equation}
    
    We first establish the sign of $x_k'$. 
    Assume for the sake of contradiction that $x_k' < 0$. By the definition of $k$, we know $\left| \frac{x_j'}{x_j} \right| \le \left| \frac{x_k'}{x_k} \right|$. It follows that:
    \begin{equation*}
        \sum_{j \neq k} d_{kj} \frac{x_j'}{x_j} \le \sum_{j \neq k} d_{kj} \left| \frac{x_k'}{x_k} \right|.
    \end{equation*}
    Substituting this bound into \cref{eq:proof_step1} yields:
    \begin{equation*}
        \frac{x_k'}{x_k} \bigg( x_k + \sum_{j \neq k} d_{kj} + 4 d_{kk} \bigg) + \sum_{j \neq k} d_{kj} \left| \frac{x_k'}{x_k} \right| \ge \delta_{1k}.
    \end{equation*}
    Since we assumed $x_k' < 0$, we have $\frac{x_k'}{x_k} = - \left| \frac{x_k'}{x_k} \right|$. The terms involving $\sum_{j \neq k} d_{kj}$ cancel out, simplifying the inequality to:
    \begin{equation*}
        \frac{x_k'}{x_k} \left( x_k + 4 d_{kk} \right) \ge \delta_{1k}.
    \end{equation*}
    The left-hand side is strictly negative (as $x_k>0$ and $d_{kk} \ge 0$), while the right-hand side is non-negative ($\delta_{1k} \in \{0, 1\}$). This is a contradiction. Therefore, $x_k'$ must be non-negative.
    
    Having established $x_k' \ge 0$, we now derive the upper bound. Rearranging \cref{eq:proof_step1}:
    \begin{equation*}
        \frac{x_k'}{x_k} \bigg( x_k + \sum_{j \neq k} d_{kj} + 4 d_{kk} \bigg) = \delta_{1k} - \sum_{j \neq k} d_{kj} \frac{x_j'}{x_j}.
    \end{equation*}
    Taking the absolute value and applying the triangle inequality:
    \begin{align*}
        \left| \frac{x_k'}{x_k} \right| \bigg( x_k + \sum_{j \neq k} d_{kj} + 4 d_{kk} \bigg) \le \delta_{1k} + \bigg| \sum_{j \neq k} d_{kj} \frac{x_j'}{x_j} \bigg| 
        \le \delta_{1k} + \sum_{j \neq k} d_{kj} \bigg| \frac{x_k'}{x_k} \bigg|.
    \end{align*}
    Subtracting $\sum_{j \neq k} d_{kj} |\frac{x_k'}{x_k}|$ from both sides yields:
    \begin{equation*}
        \left| \frac{x_k'}{x_k} \right| \left( x_k + 4 d_{kk} \right) \le \delta_{1k}.
    \end{equation*}
    Since $x_k > 0$ and $d_{kk} \ge 0$, the term $(x_k + 4d_{kk})$ is strictly positive.
    If $k \neq 1$, then $\delta_{1k} = 0$, which implies $|x_k'/x_k| \le 0 \implies x_k' = 0$. However, since $k$ is the index of maximum sensitivity, this would imply all sensitivities are zero, which contradicts the existence of the perturbation~$\sigma$.
    Therefore, the maximum relative sensitivity must occur for the input species, i.e., $k=1$.
    
    Substituting $k=1$ back into the inequality gives $x_1' \left( x_1 + 4 d_{11} \right) \le x_1$.
    Solving for the sensitivity $x_1'$:
    \begin{equation*}\label{eq:bound_final}
        x_1' \le \frac{1}{1 + 4 \frac{d_{11}}{x_1}} \le 1.
    \end{equation*}
    Thus, the concentration of the free monomer species $X_1$ cannot increase more than the concentration of the added signal $\sigma$.
\end{appendixproof}

\subsection{System-Wide Attenuation of Equilibrium Responses}

Having established that $|\frac{x_1'}{x_1}|$ is the global maximum for relative sensitivity, we now analyze the network-wide response.
At the fixed equilibrium under consideration, keep the input species indexed as $1$ and re-index the remaining species, breaking ties arbitrarily, so that the relative sensitivities are ordered by magnitude:
\begin{equation}\label{eq:ordering_assumption}
    \left| \frac{x_1'}{x_1} \right| \ge \left| \frac{x_2'}{x_2} \right| \ge \dots \ge \left| \frac{x_n'}{x_n} \right|.
\end{equation}

Now, let $d_{ij}' = \frac{dd_{ij}}{d\sigma}$ denote the sensitivity of the dimer $D_{ij}$. 

\begin{lemma}[Sign Coherence]
\label{lemma:sign_coherence}
    Given the ordering in \cref{eq:ordering_assumption}, for any dimer species $D_{ij}$ with $i < j$, the sign of the sensitivity $d_{ij}'$ is determined by the monomer species with the lower index (higher sensitivity). Specifically, if $d_{ij}' \neq 0$, then $\operatorname{sgn}(d_{ij}') = \operatorname{sgn}(x_i')$.
\end{lemma}
\begin{proof}
    From \cref{eq:eq_conc}, $d_{ij}' = d_{ij} (\frac{x_i'}{x_i} + \frac{x_j'}{x_j})$. Since $i < j$, our ordering implies $|\frac{x_i'}{x_i}| \ge |\frac{x_j'}{x_j}|$. Thus, the first term weakly dominates the sum, and the sign of $d_{ij}'$ follows the sign of $x_i'$. 
    In the symmetric case where $\frac{x_i'}{x_i} = - \frac{x_j'}{x_j}$, the terms cancel, yielding $d_{ij}'=0$.
\end{proof}

We can now rewrite the derivative of the mass conservation equation (\cref{eq:mass_conservation}) for species $i$ by splitting the sum into terms with indices $j>i$ and $j<i$:
\begin{equation}\label{eq:main-eq}
    x_i' + \sum_{j > i} d_{ij}' + 2 d_{ii}' = \delta_{1i} - \sum_{j < i} d_{ji}'.
\end{equation}
By \cref{lemma:sign_coherence}, for all $j > i$, $d_{ij}'$ has the same sign as $x_i'$ or is zero. Similarly, $d_{ii}'$ clearly has the same sign as $x_i'$. Therefore, all non-zero terms on the left-hand side of \cref{eq:main-eq} share the same sign. This allows us to take the absolute value of the sum as the sum of absolute values:
\begin{equation}\label{ineq:main}
    |x_i'| + \sum_{j > i} |d_{ij}'| + 2 |d_{ii}'| = \Big| \delta_{1i} - \sum_{j < i} d_{ji}' \Big| \le \delta_{1i} + \sum_{j < i} |d_{ji}'|.
\end{equation}

This inequality bounds the changes in monomer species $i$ and its dimers with higher-indexed monomers ($j>i$) by the changes in its dimers with lower-indexed monomers ($j<i$).
We now generalize this to show system-wide attenuation.

\begin{theorem}[No-Go Theorem for Dimerization Amplification]
\label{thm:suppression}
    For any monomer species $X_i$ in the ordered dimerization network, the sensitivities of
    $X_i$ and its dimers $D_{ij}$ with $j \geq i$ satisfy: 
    \begin{equation}\label{ineq:supp}
        |x_i'| + \sum_{j > i} |d_{ij}'| + 2 |d_{ii}'| \le 1 - \sum_{j < i} \bigg( |x_j'| + 2|d_{jj}'| + \sum_{k > i} |d_{jk}'| \bigg).
    \end{equation}
\end{theorem}
Here, the left-hand side represents the total concentration sensitivity directly involving monomer~$i$, including its free form and its dimers with higher-indexed species ($j \ge i$). The right-hand side shows that this quantity is bounded by the remaining portion of the unit input sensitivity after the lower-indexed terms are accounted for; the proof of \cref{thm:suppression} is given in \cref{proof:suppression}.
The main takeaway of the theorem is that the sensitivities of both monomer species and dimers are individually bounded: $|x_i'| \le 1$ and $|d_{ij}'| \le 1$ for all $i,j$.

\begin{appendixproof}[Proof of \cref{thm:suppression}]
\label{proof:suppression}
    To prove the theorem, we must bound the term $\sum_{j < i} |d_{ji}'|$. Applying \cref{ineq:main} to each lower-indexed species $j<i$ and isolating the term $d_{ji}'$ gives:
    \begin{equation*}
        |d_{ji}'| \le \delta_{1j} + \sum_{k < j} |d_{kj}'| - |x_j'| - \sum_{\substack{k > j \\ k \neq i}} |d_{jk}'| - 2|d_{jj}'|.
    \end{equation*}
    Summing this over all $j < i$:
    \begin{equation}\label{ineq:minor}
        \sum_{j < i} |d_{ji}'| \le \sum_{j < i} \bigg( \delta_{1j} + \sum_{k < j} |d_{kj}'| - |x_j'| - \sum_{\substack{k > j \\ k \neq i}} |d_{jk}'| - 2|d_{jj}'| \bigg).
    \end{equation}
    
    We now simplify the double summations over the interaction terms $d$. The expression contains two sums: one over $k < j$ (complexes made of lower-indexed species) and one over $k > j$ (complexes made of higher-indexed species).
    \begin{equation*}
        \sum_{j < i} \bigg( \sum_{k < j} |d_{kj}'| - \sum_{\substack{k > j \\ k \neq i}} |d_{jk}'| \bigg) = - \sum_{j < i} \sum_{k > i} |d_{jk}'|.
    \end{equation*}
    This equality holds because for any pair $\{u, v\}$ both less than $i$, the term $|d_{uv}'|$ appears once positively (when $j=v, k=u$) and once negatively (when $j=u, k=v$). These internal terms cancel out, leaving only the negative terms where $k > i$.
    
    Substituting this simplification back into \cref{ineq:minor}:
    \begin{equation}\label{ineq:minor2}
        \sum_{j < i} |d_{ji}'| \le \sum_{j < i} \bigg( \delta_{1j} - |x_j'| - 2|d_{jj}'| - \sum_{k > i} |d_{jk}'| \bigg).
    \end{equation}
    
    Finally, we merge \cref{ineq:minor2} into \cref{ineq:main}:
    \begin{equation*}
        |x_i'| + \sum_{j > i} |d_{ij}'| + 2 |d_{ii}'| \le \delta_{1i} + \sum_{j < i} \delta_{1j} - \sum_{j < i} \bigg( |x_j'| + 2|d_{jj}'| + \sum_{k > i} |d_{jk}'| \bigg).
    \end{equation*}
    Since $\sum_{j \le i} \delta_{1j} = 1$ (as the sum includes the only perturbed species $X_1$), we arrive at the inequality stated in \cref{ineq:supp}.
\end{appendixproof}

\begin{corollary}[Total Variation Bound]
    The sum of the absolute sensitivities of the monomer species is bounded by unity;
    homodimer formation strictly tightens this bound.
    \begin{equation}
        \sum_{i=1}^n |x_i'| \le 1 - 2\sum_{i=1}^n |d_{ii}'|.
    \end{equation}
\end{corollary}

\begin{proof}
    This follows directly by substituting $i=n$ into \cref{ineq:supp}.
\end{proof}

This corollary can be thought to strengthen \cref{thm:input-domination} in two ways. First, \cref{thm:input-domination} bounds the relative sensitivities of non-input monomers, but a small relative sensitivity could still correspond to a large absolute concentration change for an abundant species. The total variation bound rules this out. Second, even if no single monomer changes by more than the input, one might imagine a large distributed response spread across many monomer species. The corollary rules this out as well.

\section{Breaking the Limit: Trimerization-Based Amplification}\label{sec:trimerization}

In \cref{sec:dimerization}, we established a fundamental structural constraint: reaction networks limited to dimerization cannot amplify an input signal, regardless of the underlying thermodynamic parameters. In this section, we show that by relaxing the constraint to allow for the formation of higher-order complexes (specifically, only trimers are needed), signal amplification becomes realizable.
We propose concrete trimerization-based architectures, and present wet-lab experimental data showing amplification with a trimeric amplifier design.

\subsection{Entropy-Driven Trimer-Based Amplifier}\label{sec:entropy-driven}

\begin{figure}[t]
    \centering
    \includegraphics[width=0.58\textwidth]{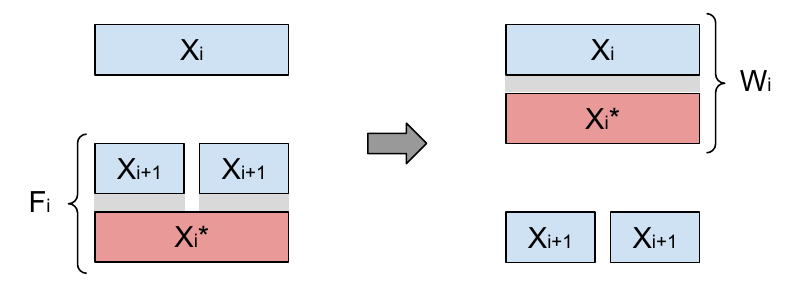}
    \caption{Reaction schematic of the entropy-driven amplifier.}
    \label{fig:entropy-driven}
\end{figure}

A simple construction for a trimer-based amplifier is the entropy-driven design illustrated in \cref{fig:entropy-driven}. This system leverages the thermodynamic drive toward greater global entropy (in the sense of the number of separate complexes) following the introduction of the input $X_0$.

Intuitively, the system implements the following abstract reaction:
\begin{equation*}
    X_{i} + F_{i} \rightarrow W_{i} + 2 X_{i+1}
\end{equation*}
Such idealized reactions, when composed,
implement amplification exponential in the number of modules:
one copy of upstream input $X_i$ can release two copies of downstream input $X_{i+1}$, doubling at every module.\footnote{If higher-order complexes (larger than trimers) are allowed, such complexes can form by exchanging sequestered signals $X_j$ among different fuel molecules $F_i$. 
As an example, $F_i + F_{i+1} \rightleftharpoons Q + W_{i+1} $, where $Q$ is formed of one copy of $X_i^*$, one copy of $X_{i+1}$, and two copies of $X_{i+2}$.
However, in the absence of input, the system still remains trapped in an equilibrium distribution without the final product $X_n$.}
Thus idealized, 
the overall stoichiometric transformation from initial reactants to final products is given by
\begin{equation*}
    X_0 + \sum_{i=0}^{n-1} 2^{i} \, F_{i} \rightarrow \sum_{i=0}^{n-1} 2^{i} \, W_{i} + 2^n X_n
\end{equation*}
and the theoretical (maximum) amplification factor of this network is $\beta = 2^n$.

There are two fundamental problems with this idealized analysis that limit amplification.
The first is the obvious structural limitation:
If the system is implemented with DNA, then, as described, each signal $X_i$ is half the length of the previous one.\footnote{Because each step in the cascade effectively ``divides'' the original input to generate multiple outputs, the amplification factor is limited by the total length of the input in base pairs $|X_0|$ and the minimum domain size~$\kappa$---the minimum strand length necessary for reliable molecular recognition and binding: $\beta \le \frac{|X_0|}{\kappa}$.}
The second, deeper limitation is based on general thermodynamic considerations, follows from \Cref{sec:thermo_bounds}, and will be explored later.

In the following section, we overcome the structural limitation by introducing an \emph{isometric} amplifier---a design that preserves the size of the output relative to the input.

\subsection{Isometric Trimer-Based Amplifier}\label{sec:isometric}

To overcome the obvious scaling limitation of the above entropy-driven cascades, we introduce a minimal---still with maximum complex size three---network designed to amplify an input signal while strictly preserving the size of the output. 
This \emph{isometric} amplification relies on a carefully engineered partial ordering of binding affinities among the system's underlying complexes.

\begin{figure}[t]
    \centering
    \includegraphics[width=0.4\textwidth]{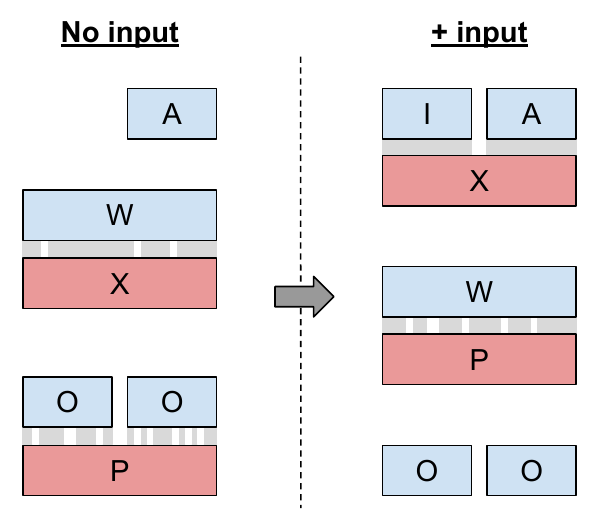}
    \caption{Schematic of the proposed isometric trimerization-based amplifier. \textbf{(Left)} The thermodynamically favored configuration in the absence of the input signal I. Strand W exhibits a higher binding affinity for X than for P. Consequently, the output strands O remain sequestered in a stable trimeric complex with the ``locking'' molecule P, forming O:O:P. 
    \textbf{(Right)} The thermodynamically favored configuration following the introduction of the input. Input I cooperatively binds with ``helper'' species A and X to form the I:A:X complex, displacing W from the W:X complex. The newly liberated strand W then drives the dissociation of the O:O:P complex by binding to P (forming W:P), thereby releasing the free output O. The forward directionality of this cascade is governed by the thermodynamic hierarchy of the complex affinities: $K_{\text{I:A:X}} > K_{\text{W:X}} > K_{\text{W:P}} > K_{\text{O:O:P}}$. Crucially, in the absence of I, species A possesses insufficient affinity to displace X or P. \textit{Note: An increased density of white gaps within the gray binding regions visually denotes progressively weaker binding affinities.}}
    \label{fig:trimer-model}
\end{figure}

\cref{fig:trimer-model} illustrates the proposed reaction pathway. 
Unlike the entropy-driven amplification scheme above, where the output species is smaller than the input, our isometric design generates an output molecule O of comparable size to the input I. This size preservation is a critical feature for building deeper, multi-stage cascading amplifiers so that the output of one stage can serve as an input for the subsequent stage.
Thermodynamically, the system is designed such that enough energy is released upon the formation of the input complex I:A:X to effectively ``pay'' the energetic penalty required to dissociate the output complex O:O:P.

In order to avoid leak (release of O without I), we must limit the affinities in the affinity hierarchy noted in the figure caption. 
Specifically, without input I, molecule A can displace half of the W:X complex. 
Then this free half of W may displace an output molecule O from the O:O:P complex, i.e., A + W:X + O:O:P $\rightarrow$ A:W:X:O:P + O.\footnote{While our theoretical model restricts complex size to three, mitigating the formation of higher-multimericity off-model polymers remains a practical engineering goal.}
This reaction pathway is inhibited because it reduces the number of separate complexes by one, but we need to ensure that the increased bonding does not compensate for the loss of entropy.

By strategically engineering sequence mismatches, we can tune the free energies of the complexes, establishing the required partial order of binding affinities and their relationship to separate complex formation entropy. 
Furthermore, because of partial complementarity, the sequence of O can be sufficiently different from I that several identical isometric amplification modules could be composed, where the output of one is input to another.

\paragraph*{Experimental Results} 

To validate the theoretical framework of the isometric trimer-based amplifier, we implemented an experimental model using synthetic DNA. 
First, we verified the theoretical performance of our designed sequences using NUPACK~\cite{NUPACK2025}. 
As observed in \cref{fig:isometric}a, while the simulated system exhibited some leakage (release of output in the absence of input) at a higher temperature (28°C), this leak was effectively eliminated by decreasing the temperature.
Overall, the NUPACK-simulated response closely follows the ideal $2\times$ slope until reaching saturation.

The experimental data shown in \Cref{fig:isometric}(a)-(b) was collected at 4°C,
but in contrast to NUPACK predictions, it still exhibits non-negligible leakage in the absence of the input (0~nM). 
This minor discrepancy may stem from several factors: imprecise stoichiometry of the added components, DNA synthesis impurities (as only the fluorophore-tagged output strand was HPLC-purified, while the rest were purified via standard PAGE), or fluorophore-induced energetic perturbations not captured by NUPACK.

\begin{figure}[t] 
    \centering 
    \includegraphics[width=0.9\textwidth]{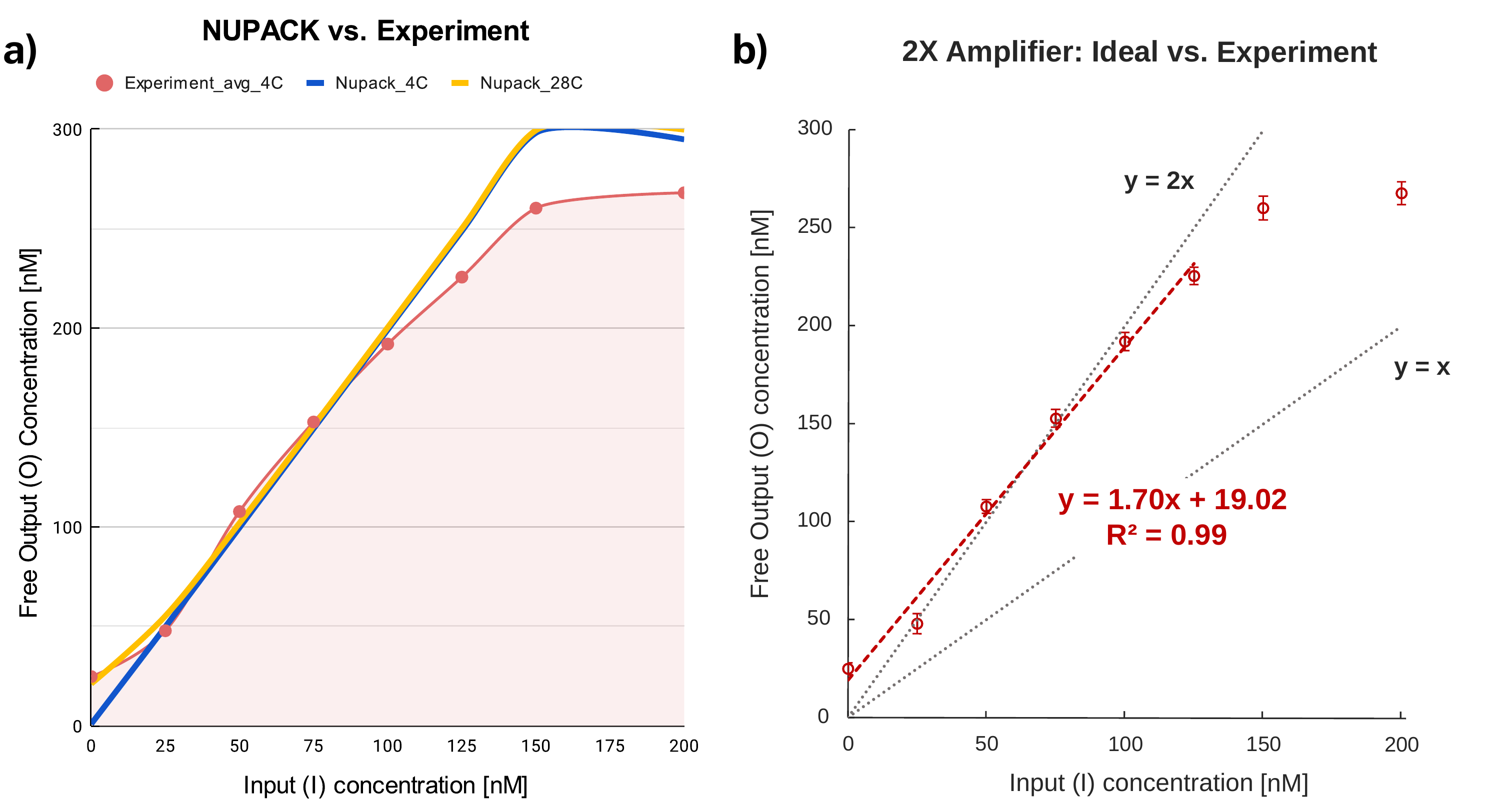} 
    \caption{Experimental results of the isometric amplifier. The input concentration is varied from 0 to 200~nM, the total concentration of the output strand is 300~nM, and all other circuit strands are set at 150~nM.  
    \textbf{(a)} Experimental data plotted against NUPACK estimated behavior.  
    \textbf{(b)}~Experimental data plotted against the ideal $2\times$ response. Linear regression of the response regime yields an experimental amplification factor of $1.7\times$.} 
    \label{fig:isometric} 
\end{figure}

As shown in \cref{fig:isometric}b, to quantify the experimental amplification factor, we performed linear regression.
While the theoretical maximum amplification factor of this specific system is $2\times$, our experimental data demonstrates $1.7\times$ signal amplification.
The regression analysis excluded the saturated data points (input concentrations of 150 nM and 200 nM) to capture the true linear amplification regime. 
(If the leakage point is also excluded from the fit to strictly evaluate the active response slope, the calculated amplification factor rises to $1.77\times$.)
Complete details of the experimental protocols, DNA sequences, and statistical methods are provided in~\cref{sec:methods}.

Since there is no physical limit imposed by shrinking molecule sizes as in the naive construction of \cref{sec:entropy-driven}, one might intuitively expect that arbitrary signal amplification is physically realizable. 
However, escaping structural limitations does not grant unlimited amplification. 
In the following section, we establish that regardless of a system's specific structural architecture or local mechanisms, there exists a fundamental, inescapable thermodynamic limit on signal amplification for \emph{all} thermodynamic equilibrium-based networks.

\section{Thermodynamic Bounds on Amplification}\label{sec:thermo_bounds}

In the previous sections, we demonstrated that fundamental structural differences (dimerization versus trimerization) can disallow or permit amplification.
We have also seen that the trimeric amplifier could in principle be composed to seemingly increase its amplification factor. 
In this section, we turn to general and unavoidable limits on amplification at thermodynamic equilibrium determined by thermodynamic considerations, regardless of structural assumptions.
Unlike prior thermodynamic bounds on nonequilibrium biochemical response~\cite{Qian2003ThermoKinetic,OwenGingrichHorowitz2020,OwenTallaBiddleGunawardena2023}, the systems considered here have no persistent chemical-potential driving: both the pre-input and post-input configurations are equilibria.
The relevant resource is therefore the binding free energy gained when the added input analyte interacts with amplifier components.

There is a strong intuition that chemical amplification should require significant energy expenditure since the system must traverse substantial ``configurational distance'' from a state with little output to one with a lot of output.
Since we are requiring the system before and after the addition of input to be at thermodynamic equilibrium, the only source of this energy must be from the input itself.
This section formalizes this reasoning, showing that the free energy with which the input binds amplifier components must scale linearly with the amplification factor. 
In the context of DNA nanotechnology, this implies that for an analyte (DNA molecule) of a fixed length, there is only so much amplification that can be realized.
Moreover, we show that the achievable amplification is further discounted by the \emph{output contrast}, defined as the ratio of the triggered output (signal) to the untriggered output (leak). 
For a fixed input free energy, demanding a sharper contrast leaves less room for amplification.

\subsection{Model}\label{sec:model}

We represent the monomer composition of any complex as a vector tuple $(u,v) \in \mathbb{N}^m \times \mathbb{N}$,
where $u$ is the vector of stoichiometric counts of non-input monomers, and $v$ is the count of the input monomer $I$.
Given a complex $(u,v)$, we define its \emph{residual core} to be just its non-input monomers.
The $L_1$-norm $|u| = \sum_{i=1}^m u_i$ of $(u,v)$ is the size of its residual core. Similarly, we define the \emph{input core} of $(u,v)$ as its constituent input monomers, with its size naturally given by $v$.

The concentration of a specific complex is denoted $x_{u,v}$. This vector representation naturally partitions the network into three distinct pools: pure residual complexes ($v=0, u \neq \mathbf{0}$), pure input complexes ($u=\mathbf{0}, v \ge 1$), and mixed complexes ($u \neq \mathbf{0}, v \ge 1$).\footnote{There is no complex $(u=\mathbf{0},v=0)$.}
All concentrations are dimensionless mole fractions---the ratio of complexes to solvent molecules (i.e., water)---and are thus smaller than one.
(The derivation of the energy function (\cref{eq:energy_function}) only
holds in the regime of substantially more dilute complexes than solvent~\cite{DBLP:journals/siam/DirksBSWP07}.)

To model the equilibrium behavior of these systems, we adopt the free-energy formulation utilized by Dirks \emph{et al.}~\cite{DBLP:journals/siam/DirksBSWP07}.
Let $g(x)$ denote the dimensionless free energy of the overall system configuration (formula given below), let $G_{u,v}$ represent the dimensionless free energy of a specific complex, and let $\mathcal{P}$ denote the set of all valid complexes in the network. 
We define the sign convention for $G_{u,v}$ such that more negative values correspond to more thermodynamically favorable complexes.

Let $\sigma$ represent the total concentration of the input monomer $I$.
The equilibrium concentrations are obtained by minimizing the system's total free energy (corresponding to the pseudo-Helmholtz free energy utilized extensively throughout the chemical reaction network theory literature~\cite{horn1972general,feinberg2019foundations}):\footnote{In this work, $\log$ denotes the natural logarithm.}
\begin{equation}\label{eq:energy_function}
    g(x) = \sum_{(u,v) \in \mathcal{P}} x_{u,v} (\log x_{u,v} - 1 + G_{u,v})
\end{equation}
This minimization is subject to the mass conservation constraint $A x = x^0(\sigma)$, where $x$ is the vector of complex concentrations, $x^0(\sigma)$ is the vector representing the total concentrations of all monomers (including the amount $\sigma$ of added input $I$), and $A$ is the mass conservation matrix mapping the concentrations of complexes to their constituent monomers.

To formally manipulate the system configurations, we define two projection operators that isolate the cores from the overall complex concentration vector $x$. 

The \textbf{Input-Core operator}, denoted $\ic[x]$, extracts the input cores by stripping away all residual components. 
In other words, it removes all residual monomers, and for every $v \geq 1$, creates a complex containing $v$ input monomers whose concentration is the sum of the concentrations of all complexes containing exactly $v$ input monomers in $x$. Formally, we define the resulting vector $w = \ic[x]$:
\begin{equation} \label{eq:ic}
    w = \ic[x] \implies 
    \begin{cases} 
        w_{(\mathbf{0},v)} = \displaystyle \sum_{u} x_{(u,v)} \\[3ex]
        w_{(u,v)} = 0 & \text{for } u \neq \mathbf{0}
    \end{cases}
\end{equation}

Conversely, the \textbf{Residual-Core operator}, denoted $\rc[x]$, extracts the residual cores by stripping away all input monomers. Analogously to \Cref{eq:ic}, we define the resulting vector $w = \rc[x]$:
\begin{equation}
    w = \rc[x] \implies 
    \begin{cases} 
        w_{(u,0)} = \displaystyle \sum_{v} x_{(u,v)} \\[3ex]
        w_{(u,v)} = 0 & \text{for } v \ge 1
    \end{cases}
\end{equation}

\subsection{Energy Change upon Input Addition}\label{sec:energy-change}

We now investigate the total change in system free energy upon the introduction of the input. 
To quantify this transition, we compare the thermodynamic equilibria of the network before and after the input is introduced---configurations we formalize as the \emph{Resting} and \emph{Active} states, respectively. 
To establish a rigorous mathematical bound on this free-energy difference, we will also define a specific configuration termed the \emph{Drained} state. 

The Active state is the thermodynamic equilibrium of the entire system with concentration $\sigma$ of the input $I$.
Formally:
\begin{definition}\label{def:active}
    For any input concentration $\sigma \geq 0$, the
     \textbf{Active state} $\xact$ is the concentration vector that minimizes the total system free energy $g(x)$ subject to mass conservation:
    \begin{equation}
        \begin{aligned}
            \xact = \operatorname{argmin}_{x \ge 0} \quad & g(x) \\
            \text{s.t.} \quad & A x = x^0(\sigma)
        \end{aligned}
    \end{equation}
\end{definition}

Next we define the Drained state, which represents ``ripping'' the input out of the Active state without allowing the system to relax to a new equilibrium.
Formally, the Drained state is constructed by separating the input cores and residual cores of the Active state:
\begin{definition}\label{def:drained}
    \textbf{Drained state} $\xdrn$ is defined by:
    \begin{equation}
        \xdrn = \ic[\xact] + \rc[\xact]
    \end{equation}
\end{definition}
Note that the dissociation into input cores and residual cores preserves mass conservation $A \xdrn = x^0(\sigma)$. 

Finally, we define the Resting state, which represents the thermodynamic equilibrium without input to which the input is added without being able to interact with the rest of the system.
More precisely, the Resting state is the thermodynamic equilibrium of the system without the input, plus $\ic[\xact]$.
Formally:
\begin{definition}\label{def:resting}
    \textbf{Resting state} $\xrst$ is defined by:
    \begin{equation}
        \begin{aligned}
            \xrst = \ic[\xact] + \operatorname{argmin}_{x \ge 0} \quad & g(x) \\
            \text{s.t.} \quad & A x = x^0(0) 
        \end{aligned}
    \end{equation}
\end{definition}

\begin{lemma}[Free Energy Hierarchy]
\label{lem:energy_hierarchy}
    The free energies of the Active, Resting, and Drained states satisfy the following inequality:
    \begin{equation*}
        g(\xact) \le g(\xrst) \le g(\xdrn)
    \end{equation*}
\end{lemma}

The proof of \cref{lem:energy_hierarchy} is given in \cref{proof:lemma_hierarchy}.

\begin{appendixproof}[Proof of \cref{lem:energy_hierarchy}]
\label{proof:lemma_hierarchy}
\text{}
    First, we evaluate the relationship between the Active and Resting states. By definition, $\xact$ is the absolute global minimizer of the free energy $g(x)$ over the entire feasible concentration space subject to the total mass constraint $A x = x^0(\sigma)$. The Resting state, defined as $\xrst = \ic[\xact] + \operatorname{argmin}_{Ax=x^0(0)} g(x)$, is also a physically valid configuration in this same feasible space; its decoupled input mass and minimized residual mass exactly sum to the total network mass $x^0(\sigma)$. Because $\xact$ is the unconstrained minimizer over the entire valid space, and $\xrst$ represents just one specific state within that space (specifically, one restricted to have no mixed complexes), it necessarily follows that $g(\xact) \le g(\xrst)$.

    Second, we compare the Resting and Drained states. Based on \cref{def:drained} and \cref{def:resting}, both states share the exact same isolated input pool, $\ic[\xact]$. Therefore, any difference in their total free energies depends entirely on their respective residual pools. 
    
    The residual pool of the Drained state is given by the projection $\rc[\xact]$. Because this projection conserves all residual monomers from the Active state, it perfectly satisfies the isolated residual mass constraint $A \cdot \rc[\xact] = x^0(0)$. Thus, $\rc[\xact]$ is a valid concentration vector in the isolated residual space. However, the residual component of the Resting state is explicitly defined as the global free energy minimizer over this exact same isolated space (i.e., $\operatorname{argmin}_{Ax=x^0(0)} g(x)$). Because the Resting state pairs the shared input pool with the thermodynamically optimal residual configuration, whereas the Drained state pairs it with the unoptimized active projection, we must conclude that $g(\xrst) \le g(\xdrn)$.
\end{appendixproof}

We now define the following quantities with respect to the Active state $\xact$.
Let $\mathcal{U}=~\{u \neq~\mathbf{0} \mid \exists v \ge 1 \colon \xact_{u,v} > 0\}$ denote the set of residual stoichiometries $u$ that can form complexes with the input in the Active state.
Let $X_{\text{inp}} = \sum_{v\ge 1} \xact_{\mathbf{0},v}$ be the total concentration of the pure input complexes, $X_{\mathrm{res}} =\sum_{u \in \mathcal{U}} \xact_{u,0}$ be the total concentration of the pure residual complexes, and let $X_{\mathrm{mix}}=\sum_{u \in \mathcal{U}} \sum_{v\ge 1} \xact_{u,v}$ be the total concentration of the mixed complexes.

Next, we define the binding free energy, $\Delta G_{u,v}^{\text{bind}} := G_{u,v} - G_{u,0} - G_{\mathbf{0},v}$, which quantifies the thermodynamic cost of forming the complex $(u,v)$ from its isolated residual and input cores. 
Finally, we let $\delGinpstg = \min_{u,v} \Delta G_{u,v}^{\text{bind}}$ denote the most favorable binding energy achievable within the network upon interaction with the input.
Note that we use the convention that $\delGinpstg < 0$, and the more negative it is, the more favorable the interaction.

Having formally defined the Active, Drained, and Resting states, our primary objective is to quantify the total change in the pseudo-Helmholtz free energy of the system upon the addition of an input species with concentration $\sigma$. 
The Drained state serves as a crucial analytical intermediate. Starting from the Active state, where the input has already propagated through the network to trigger the targeted outputs, we conceptually decouple all bound input cores from the residual network. 
If we then allow the residual cores of this Drained state to thermodynamically relax while remaining strictly isolated from the input pool, the network equilibrates into the Resting state. 

Altogether, by comparing the Active, Drained, and Resting states, and incorporating the strongest potential energy released by input binding ($\delGinpstg$), we can rigorously bound the total pseudo-Helmholtz free energy change of the system as shown below. 
We show that the bound depends on the following dimensionless parameter, which we term amplification capacity.
Intuitively, the upper bound on amplification capacity is the maximum energy with which the input binds ($-\delGinpstg$). 
From this upper bound, we subtract the log of ``concentration of amplifier components'' (as captured by $X_{\mathrm{res}} + X_{\mathrm{mix}}$)---with smaller concentrations leading to smaller amplification capacity.\footnote{Recall that all concentrations are strictly less than one, as they are expressed as dimensionless mole fractions (the ratio of complexes to solvent molecules). Consequently, $X_{\mathrm{res}} + X_{\mathrm{mix}} < 1$.}

\begin{definition}
    We define the \emph{amplification potential} as $\psi(X, \Delta G) := -\Delta G + \log (X)$. Hence, the input amplification potential is denoted by $\psi_{\mathrm{inp}} := \psi(X_{\mathrm{res}} + X_{\mathrm{mix}}, \delGinpstg)$. We then define the \emph{amplification capacity} $\phi$ as follows:
    \begin{equation}\label{eq:Phi}
        \phi := 
        \begin{cases}
            e^{\psi_{\mathrm{inp}}} & \text{if } \psi_{\mathrm{inp}} \le 0, \\
            \psi_{\mathrm{inp}} + 1 & \text{if } \psi_{\mathrm{inp}} > 0.
        \end{cases}
    \end{equation}
\end{definition}
Below we will derive upper bounds on the amplification factor as a function of the amplification capacity, justifying this terminology.

\begin{theorem}[Input-Driven Energy Change]\label{thm:energy_bound}
    Let $\sigma$ be the total added concentration of the input species. Let $g(\xdrn)$ denote the free energy of the Drained state, let $g(\xact)$ denote the free energy of the Active state. 
    Then, the amplification capacity $\phi$ satisfies:   \begin{equation}\label{eq:upper_bound_final}
        g(\xdrn) - g(\xact) \le
        \sigma \, \phi
    \end{equation}
\end{theorem}
The proof of \cref{thm:energy_bound} is given in \cref{proof:energy-bound}.

\begin{corollary}\label{rem:energy_hierarchy_bounds}
    A consequence of the free energy hierarchy (\Cref{lem:energy_hierarchy}) is that the intermediate free energy transitions are likewise bounded. Specifically, because $g(\xact) \le g(\xrst) \le g(\xdrn)$, it holds that both $g(\xrst) - g(\xact) \le \sigma \, \phi$ and $g(\xdrn) - g(\xrst) \le \sigma \, \phi$.
\end{corollary}

\begin{appendixproof}[Proof of \cref{thm:energy_bound}]
\label{proof:energy-bound}
\text{}
\paragraph*{Step 1: Expressing $g(\xdrn) - g(\xact)$.}

Expanding the free energy function for the Drained and Active states yields:
\begin{align*}
    g(\xdrn) - g(\xact) &= \sum_{v\ge 1} \Big[ \xdrn_{\mathbf{0},v} \big( \log \xdrn_{\mathbf{0},v} - 1 + G_{\mathbf{0},v} \big) - \xact_{\mathbf{0},v} \big( \log \xact_{\mathbf{0},v} - 1 + G_{\mathbf{0},v} \big) \Big] \\
    &\quad + \sum_{u \neq \mathbf{0}} \Bigg[ \xdrn_{u,0} \big( \log \xdrn_{u,0} - 1 + G_{u,0} \big) - \sum_{v} \xact_{u,v} \big( \log \xact_{u,v} - 1 + G_{u,v} \big) \Bigg]
\end{align*}

By definition, the Drained state uses $\rc$ and $\ic$ operators mapping all mixed complexes onto their respective pure residual and pure input pools. Thus, for any pure residual core $u$ (where $u \neq \mathbf{0}$), its concentration in the Drained state is exactly the sum over all its possible input components $v$, which means $\xdrn_{u,0} = \sum_{v} \xact_{u,v}$.

Since all extracted input cores get mapped onto the pure input pool complexes of the same size, the concentration of pure input complexes of size $v$ (where $v \ge 1$) becomes $\xdrn_{\mathbf{0},v} = \sum_u \xact_{u,v}$.

Substituting these distributions into the free energy function $g(x)$ and grouping the terms by their input cores $v$ and residual cores $u$, the energy difference resolves to:
\begin{align*}
    g(\xdrn) - g(\xact) &= \sum_{v\ge 1} \Big[ \sum_u \xact_{u,v} \big( \log \sum_u \xact_{u,v} - 1 + G_{\mathbf{0},v} \big) - \xact_{\mathbf{0},v} \big( \log \xact_{\mathbf{0},v} - 1 + G_{\mathbf{0},v} \big) \Big] \\
    &\quad + \sum_{u \neq \mathbf{0}} \Bigg[ \sum_{v} \xact_{u,v} \big( \log \sum_{v} \xact_{u,v} - 1 + G_{u,0} \big) - \sum_{v} \xact_{u,v} \big( \log \xact_{u,v} - 1 + G_{u,v} \big) \Bigg]
\end{align*}

By expanding the equations and canceling the constant $-1$ terms (noting that these terms in the second line perfectly cancel each other out, while the terms of the first line leave the total scalar concentration of all mixed complexes), the energy difference resolves to:
\begin{equation}\label{eq:act-drn}
    \begin{aligned}
        g(\xdrn) - g(\xact) &= \underbrace{\sum_{v\ge 1} \Bigg[ \Big( \sum_u \xact_{u,v} \Big) \log \Big( \sum_u \xact_{u,v} \Big) - \xact_{\mathbf{0},v} \log \xact_{\mathbf{0},v} \Bigg]}_{\mathcal{A}_{\text{inp}}} \\
        &\quad + \underbrace{\sum_{u \neq \mathbf{0}} \Bigg[ \Big( \sum_{v} \xact_{u,v} \Big) \log \Big( \sum_{v} \xact_{u,v} \Big) - \sum_{v} \xact_{u,v} \log \xact_{u,v} \Bigg]}_{\mathcal{A}_{\text{res}}} \\
        &\quad - \sum_{v\ge 1} \sum_{u \neq \mathbf{0}} \left[ \xact_{u,v} \left( 1+ \Delta G_{u,v}^{\text{bind}} \right) \right]
    \end{aligned}
\end{equation}
where the quantities $\mathcal{A}_{\text{inp}}$ and $\mathcal{A}_{\text{res}}$ capture the Shannon-like entropy of mixing for the pure input complexes and the pure residual complexes, respectively.

\paragraph*{Step 2: Establishing an Upper Bound for the Free Energy Difference.}

To establish the upper bound on the free energy difference $g(\xdrn) - g(\xact)$, we must evaluate the entropic terms $\mathcal{A}_{\text{inp}}$ and $\mathcal{A}_{\text{res}}$. We rely on a fundamental algebraic inequality for logarithms: for any $A, B > 0$, the relation $A \log(1 + B/A) \le B$ holds true. Consequently, we can bound the entropy of a sum as follows:
\begin{equation*}
    (A+B) \log(A+B) - A \log A = A \log\left(1 + \frac{B}{A}\right) + B \log(A+B) \le B + B \log(A+B)
\end{equation*}

We apply this inequality directly to the sums defining $\mathcal{A}_{\text{inp}}$ and $\mathcal{A}_{\text{res}}$. For $\mathcal{A}_{\text{inp}}$, let $A = \xact_{\mathbf{0},v}$ and $B = \sum_{ u \neq \mathbf{0}} \xact_{u,v}$, such that $A+B = \xdrn_{\mathbf{0},v}$. This yields:
\begin{equation*}
    \mathcal{A}_{\text{inp}} \le \sum_{v\ge 1} \Bigg[ \Bigg( \sum_{u \neq \mathbf{0}} \xact_{u,v} \Bigg) \big( \log \xdrn_{\mathbf{0},v} + 1 \big) \Bigg] = \sum_{u \neq \mathbf{0}} \sum_{v\ge 1} \left[ \xact_{u,v} \big( \log \xdrn_{\mathbf{0},v} + 1 \big) \right]
\end{equation*}

Similarly, for $\mathcal{A}_{\text{res}}$, we set $A = \xact_{u,0}$ and $B = \sum_{v\ge 1} \xact_{u,v}$, yielding $A+B = \xdrn_{u,0}$. Expanding the subtracted sum over $v$, we find:
\begin{equation*}
    \mathcal{A}_{\text{res}} \le \sum_{u \neq \mathbf{0}} \Bigg[ \Bigg( \sum_{v\ge 1} \xact_{u,v} \Bigg) \big( \log \xdrn_{u,0} + 1 \big) - \sum_{v\ge 1} \xact_{u,v} \log \xact_{u,v} \Bigg] = \sum_{u \neq \mathbf{0}} \sum_{v\ge 1} \left[ \xact_{u,v} \big( \log \xdrn_{u,0} + 1 - \log \xact_{u,v} \big) \right]
\end{equation*}

Substituting these bounded entropic quantities back into \cref{eq:act-drn}, the total subtracted scalar concentration of the mixed complexes (the isolated sum in the third line) perfectly absorbs one of the $+1$ constants generated by our inequalities. The expression simplifies into a unified single summation strictly over the mixed complexes:
\begin{equation*}
\begin{aligned}
    g(\xdrn) - g(\xact) &\le \sum_{u \neq \mathbf{0}} \sum_{v\ge 1} \xact_{u,v} \Big[ \log \xdrn_{\mathbf{0},v} + 1 + \log \xdrn_{u,0} + 1 - \log \xact_{u,v} - 1 - \Delta G_{u,v}^{\text{bind}} \Big] \\
    &= \sum_{u \neq \mathbf{0}} \sum_{v\ge 1} \xact_{u,v} \Bigg[ \log \left( \frac{\xdrn_{\mathbf{0},v} \xdrn_{u,0}}{\xact_{u,v}} \right) + 1 - \Delta G_{u,v}^{\text{bind}} \Bigg]
\end{aligned}
\end{equation*}

Because each term in the preceding expression is multiplied by $\xact_{u,v}$, any complex $(u,v)$ with a zero Active-state concentration vanishes and can be excluded. Consequently, instead of summing over all $u \neq \mathbf{0}$, we can restrict the summation to $u \in \mathcal{U}$. Recall that this set is defined as: $\mathcal{U}=\{u \neq \mathbf{0} \mid \exists v \ge 1 \colon \xact_{u,v} > 0\}$.
\begin{equation}\label{ineq:fine-grained}
\begin{aligned}
    g(\xdrn) - g(\xact)
    & \le \sum_{u \in \mathcal{U}} \sum_{v\ge 1} \xact_{u,v} \Bigg[ \log \left( \frac{\xdrn_{\mathbf{0},v} \xdrn_{u,0}}{\xact_{u,v}} \right) + 1 - \Delta G_{u,v}^{\text{bind}} \Bigg]
\end{aligned}
\end{equation}

\begin{remark}\label{rm:maximizing}
    Before applying simplifications and decoupling steps to derive a closed-form tractable bound, evaluating the component-level dependencies of the preceding inequality may be appropriate to generate design strategies and effective amplifiers. 
    One such heuristic is apparent: maximize the concentration of unbound ``input partners'' ($\xact_{u,0}$) for any residual core $u$ where the corresponding mixed complex ($\xact_{u,v}$) is significant, which increases the drained concentration $\xdrn_{u,0}$ and mitigates the negative logarithmic penalty associated with complex formation.
\end{remark}

The formulation in \cref{ineq:fine-grained} casts the free energy transition strictly as an expectation over the mixed complexes. To globally decouple the specific complex concentrations, we apply the Log-Sum Inequality: $\sum a_i \log(b_i/a_i) \le (\sum a_i) \log(\sum b_i / \sum a_i)$.

For the numerator sum $\sum b_i$, we evaluate $\sum_{u \in \mathcal{U}} \sum_{v\ge 1} \xdrn_{\mathbf{0},v} \xdrn_{u,0}$. Because all concentration terms are non-negative, we can create a separable upper bound, $\sum_{v\ge 1} \xdrn_{\mathbf{0},v} \cdot \sum_{u \in \mathcal{U}} \xdrn_{u,0}$. Crucially, because the Drained state is defined as sums over the Active state complexes, summing these over their respective domains exactly partitions them into the pure and mixed Active state pools:
\begin{align*}
    \sum_{v\ge 1} \xdrn_{\mathbf{0},v} =
    \underbrace{\sum_{v\ge 1} \xact_{\mathbf{0},v}}_{X_{\text{inp}}} + \underbrace{\sum_{u \in \mathcal{U}} \sum_{v\ge 1} \xact_{u,v}}_{X_{\mathrm{mix}}} \qquad \qquad \qquad
    \sum_{u \in \mathcal{U}} \xdrn_{u,0} = 
    \underbrace{\sum_{u \in \mathcal{U}} \xact_{u,0}}_{X_{\mathrm{res}}} + \underbrace{\sum_{u \in \mathcal{U}} \sum_{v\ge 1} \xact_{u,v}}_{X_{\mathrm{mix}}}
\end{align*}

Substituting this separated numerator back into the Log-Sum framework, the upper bound resolves to a thermodynamic evaluation dependent solely on the Active state partitioning:
\begin{equation*}
    g(\xdrn) - g(\xact) \le
    X_{\mathrm{mix}} \Bigg[ \log \left( \frac{ (X_{\text{inp}} + X_{\mathrm{mix}}) \cdot (X_{\mathrm{res}} + X_{\mathrm{mix}}) }{ X_{\mathrm{mix}} } \right) + 1 - \delGinpstg \Bigg]
\end{equation*}
To make the expression easier to work with, we observe that $X_{\text{inp}} + X_{\mathrm{mix}} \le \sigma$, so:
\begin{equation*}
    g(\xdrn) - g(\xact) \le
    \underbrace{X_{\mathrm{mix}} \Bigg[ \log \left( \frac{ \sigma \, (X_{\mathrm{res}} + X_{\mathrm{mix}}) }{ X_{\mathrm{mix}} } \right) + 1 - \delGinpstg \Bigg]}_{h(X_{\mathrm{mix}})}
\end{equation*}

To find the global maximum of this bound, we treat $(X_{\mathrm{res}} + X_{\mathrm{mix}})$ as a fixed parameter. Under this condition, $h(X_{\mathrm{mix}})$ is a strictly concave function with respect to $X_{\mathrm{mix}}$. By setting the derivative to zero, the unconstrained maximum of $h$ occurs at $X_{\mathrm{mix}} = \sigma \, (X_{\mathrm{res}} + X_{\mathrm{mix}}) \, e^{-\delGinpstg} = \sigma \, e^{\psi_{\mathrm{inp}}}$. 

Because $X_{\mathrm{mix}}$ is physically constrained by the total input concentration ($X_{\mathrm{mix}} \le \sigma$), we must evaluate the boundary conditions. If $\psi_{\mathrm{inp}} \le 0$, the optimal value lies within the valid physical domain, yielding a maximum of $\sigma e^{\psi_{\mathrm{inp}}}$. If $\psi_{\mathrm{inp}} > 0$, the unconstrained optimal value exceeds $\sigma$; because $h$ is concave, the constrained maximum must occur at the boundary $X_{\mathrm{mix}} = \sigma$, which evaluates to $\sigma(\psi_{\mathrm{inp}} + 1)$. Consequently, \cref{eq:upper_bound_final} is proven:
\begin{equation*}
    g(\xdrn) - g(\xact) \le
    \begin{cases}
        \sigma \, e^{\psi_{\mathrm{inp}}} & \text{if } \psi_{\mathrm{inp}} \le 0 \\
        \sigma \left( \psi_{\mathrm{inp}} +1 \right) & \text{if } \psi_{\mathrm{inp}} > 0
    \end{cases}
    \qedhere
\end{equation*}
\end{appendixproof}

To connect these thermodynamic constraints to system performance, we now translate the bounded free energy transitions into concentrations, yielding an explicit upper bound on the amplification factor.

\subsection{Amplification Factor Upper Bound}
\label{app:input-output-bound}

Fix an input amount $\sigma$.
Intuitively, we want the output amount to be much larger than the leak, which is the output amount in the absence of input.
For actual amplification, we also want the output amount to be much larger than the input amount $\sigma$.
The bound on energy release as a function of amplification capacity reported in \Cref{thm:energy_bound} can be recast into an explicit relationship between these quantities, properly formalized.

Let $O$ denote the designated free-output species, and define its concentration at the Active and Resting states as:
\begin{equation*}
    a := \xact[O], \qquad r := \xrst[O]
\end{equation*}
In other words, $a$ is the signal amount of output when the amplifier is activated, and $r$ is the leak amount.
We define the amplification factor $\beta$ and the output contrast $\gamma$ as:
\begin{equation*}
    \beta := \frac{a-r}{\sigma},
    \qquad
    \gamma := \frac{a}{r}.
\end{equation*}
Thus the amplification factor $\beta$ represents the amount of output increase between leak and signal, relative to the input amount. 
Output contrast $\gamma$ represents the ratio of the free output species between leak and signal states.

\begin{theorem}[Thermodynamic Bound on Amplification]\label{thm:io-contrast-bound}
    Under the natural assumption for amplification that the free-output concentration in the Active state exceeds its Resting leak ($\gamma >1$), the amplification factor $\beta$ is upper-bounded by amplification capacity $\phi$ discounted by the logarithmic output contrast:
    \begin{equation}
        \beta \le \frac{\phi}{\log \gamma}.
    \end{equation}
\end{theorem}

In the regime relevant for amplification ($\psi_{\mathrm{inp}} > 0$), the amplification capacity $\phi = \psi_{\mathrm{inp}} + 1$ grows linearly with the strength of the input interaction $-\delGinpstg$. A key takeaway of this theorem is therefore that the input must bind the amplifier components with a free energy that scales (at least) linearly with the desired amplification factor (at fixed contrast).

\begin{remark}\label{rm:amp_bound}
    In the positive input potential regime ($\psi_{\mathrm{inp}} > 0$), expanding $\phi$ in \cref{thm:io-contrast-bound} reveals:
    \begin{equation}
        \beta \le \frac{-\delGinpstg + \log (X_{\mathrm{res}} + X_{\mathrm{mix}}) + 1}{ \log \gamma} < \frac{-\delGinpstg + 1}{ \log \gamma }.
    \end{equation}
\end{remark}

\begin{proof}[Proof of \cref{thm:io-contrast-bound}]

The Resting and Drained states share the same isolated input pool, so their free-energy difference is entirely a residual-subsystem difference.  The residual part of $\xrst$ is the equilibrium minimizer subject to the isolated residual mass constraint.  Hence the pseudo-Helmholtz Bregman identity gives \cite{csiszar1975divergence}:
\begin{equation*}\label{eq:io-bregman-identity}
    g(\xdrn)-g(\xrst)
    =
    \sum_i
    \left(
        \xdrn[i]\log\frac{\xdrn[i]}{\xrst[i]}
        -\xdrn[i]+\xrst[i]
    \right),
\end{equation*}
where the sum is over the pure residual complexes.  Each summand is nonnegative. 
Keeping only the output $O$:
\begin{equation*}\label{eq:io-drained-resting-output-cost}
    g(\xdrn)-g(\xrst) \ge
    \xdrn[O] \log \frac{\xdrn[O]}{\xrst[O]} - \xdrn[O] + \xrst[O].
\end{equation*}

We know the output $O$ is a pure residual species and that draining the input from the Active state may increase its concentration $\xdrn[O] \ge \xact[O]=a$.
This holds because a fraction of the total output mass might be bound with input species in the Active state. 
Letting $r=\xrst[O]$, define
\[
    F(t):=t\log\frac{t}{r}-t+r.
\]
Since $F'(t)=\log(t/r)$, the function $F$ is increasing for $t\ge r$.  Applying $\xdrn[O] \ge a$ along with $a>r$ yields:
\begin{equation}\label{eq:io-drained-resting-coordinate-cost}
    g(\xdrn)-g(\xrst)  \ge 
    a\log\frac{a}{r}-a+r.
\end{equation}

A sharper estimate is obtained by also retaining the output contribution between the Resting and Active states.  Since the Active state is the equilibrium minimizer for the full system under the same total mass constraint, the pseudo-Helmholtz Bregman identity similarly gives:
\begin{equation}\label{eq:io-resting-active-coordinate-cost}
    g(\xrst)-g(\xact)
    \ge
    r\log\frac{r}{a}-r+a.
\end{equation}

Adding \cref{eq:io-drained-resting-coordinate-cost,eq:io-resting-active-coordinate-cost} gives:
\begin{equation*}\label{eq:io-symmetric-coordinate-cost}
\begin{aligned}
    g(\xdrn)-g(\xact)
    &\ge
    \left(a\log\frac{a}{r}-a+r\right)
    +
    \left(r\log\frac{r}{a}-r+a\right)  \\
    &=
    (a-r)\log\frac{a}{r}.
\end{aligned}
\end{equation*}

Substituting $\gamma=a/r$ and $a-r=\beta\sigma$ gives:
\begin{equation*}
    g(\xdrn)-g(\xact)
    \ge
    \beta\sigma\log\gamma.
\end{equation*}

\begin{remark}\label{rm:minimize}
    One of the reasons \Cref{thm:io-contrast-bound} is not strictly tight is the truncation of terms in \cref{eq:io-drained-resting-coordinate-cost,eq:io-resting-active-coordinate-cost}.
    Nonetheless, looking at the entire sum may be appropriate to generate design strategies and effective amplifiers. 
    One such heuristic is apparent: limit the variety of distinct chemical species that do not contain the output, which decreases the number of non-zero terms in the sum.
\end{remark}

Combining this lower bound with \cref{thm:energy_bound},
\[
    g(\xdrn)-g(\xact)\le \sigma\phi,
\]
and cancelling $\sigma$ completes the proof.
\end{proof}

To illustrate the order of magnitude of these thermodynamic limits, we evaluate the upper bound using the parameters of the isometric trimeric amplifier from \cref{sec:isometric}. 
Our input strand has a maximum binding affinity of roughly $-49~\text{kcal/mol}$ at the operating temperature of $4^\circ\text{C}$ ($277.15~\text{K}$). 
Assuming an overestimate of $X_\mathrm{mix} + X_\mathrm{res} < 1~\mu\text{M}$,
we obtain the following dimensionless parameter values\footnote{
Using the molarity of water ($55.49~\text{M}$ at $4^\circ\text{C}$) as the solvent to non-dimensionalize the system, and adopting the ideal gas constant $R = 0.00198717~\text{kcal/mol}.\text{K}$.}:
\begin{equation*}
    \delGinpstg \approx -89, \quad X_{\mathrm{res}} + X_{\mathrm{mix}} \lesssim 1.8 \times 10^{-8}.
\end{equation*}
Substituting these values into the definition for the amplification capacity yields:
\begin{equation*}
    \phi = \psi_{\mathrm{inp}} + 1 = -\delGinpstg + \log (X_{\mathrm{res}} + X_{\mathrm{mix}} ) + 1 \approx 72.
\end{equation*}
Under these conditions, requiring a modest $10$-fold target output contrast ($\gamma = 10$) restricts the maximum achievable amplification factor to $\beta \lesssim 31$. This concrete example demonstrates how a seemingly vast input free-energy budget (hybridization energy of $-49~\text{kcal/mol}$ is typically considered all but irreversible) is naturally attenuated by the system's concentration scales and signal contrast requirements, leaving a ceiling on the final equilibrium gain.\footnote{As discussed in \cref{sec:conclusion}, the tightness of our results remains open. While we have not demonstrated multi-stage composition and the bound need not be tight, it leaves open the possibility that $4$ layers of the isometric trimeric amplifier could be composed, since each contributes a factor of $2$ and $2^4 = 16 \le 31 < 32 = 2^5$.}

\subsection{Output Contrast Upper Bound}

In the previous section, \Cref{thm:io-contrast-bound} bounds the amplification factor $\beta$ as a function of amplification capacity $\phi$ and the contrast $\gamma$. 
In this section, our results address the following intuition in turn: to have high contrast, the output must be capable of binding strongly with something in order to have low leak. 
Thus, we turn from the free energy of binding the input to that of binding the output.

In order to link contrast to the underlying thermodynamics of output sequestration in the Resting state, 
suppose a set of pure residual complexes $\mathcal{Q}$ sequesters the output species in the Resting state. Each complex $Q \in \mathcal{Q}$ forms via the reversible binding of one free output molecule $O$ to a distinct residual partner $P_Q$\footnote{If $Q$ includes more than one copy of $O$, then $P_Q$ includes $O$ as well, which does not cause a problem for the subsequent reasoning steps.}:
\begin{equation*}
    O + P_Q \rightleftharpoons Q
\end{equation*}
We define the sequestration free-energy gain for each complex as $\Delta G^{\mathrm{out}}_Q := G_Q - G_O - G_{P_Q}$.
Let $\delGoutstg := \min_{Q\in\mathcal{Q}} \Delta G_Q^{\mathrm{out}}$ denote the most favorable binding energy achievable within the network upon the output sequestration.
Note that similar to $\delGinpstg$, we use the convention that $\delGoutstg < 0$, and the more negative it is, the more favorable the interaction.

By the law of mass action, the Resting equilibrium configuration satisfies:
\begin{equation}\label{eq:io-output-mass-action}
    \xrst[Q] = \xrst[O]\,\xrst[P_Q]\,e^{-\Delta G^{\mathrm{out}}_Q}
\end{equation}
Let $Z_{\mathrm{res}}^{\mathrm{out}} := \sum_{Q\in\mathcal{Q}}\xrst[P_Q]$. This quantity serves as the output-side counterpart to $X_{\mathrm{res}}$: whereas $X_{\mathrm{res}}$ represents the total concentration of complexes capable of binding the input species in the Active state, $Z_{\mathrm{res}}^{\mathrm{out}}$ represents the total concentration of complexes $P_Q$ capable of sequestering the output species in the Resting state.
Furthermore, let $\nu_{\max}$ be the largest output stoichiometry of any complex $Q$.

Finally, the output potential $\psi_{\mathrm{out}}$ is defined as follows:
\begin{equation*}
    \psi_{\mathrm{out}} := \psi(\nu_{\max} \, Z_{\mathrm{res}}^{\mathrm{out}} , \delGoutstg)
    = - \delGoutstg + \log (\nu_{\max} \, Z_{\mathrm{res}}^{\mathrm{out}}).
\end{equation*}

\begin{lemma}[Contrast Upper Bound]\label{lemma:contrast_upperbound}
    The output contrast $\gamma$ of an amplifier is bounded by:
    $\gamma \le 1+ e^{\psi_{\mathrm{out}}}$.
\end{lemma}

\begin{proof}
    The net increase in the free output concentration, $a - r$, cannot exceed the total amount of sequestered output available in the Resting state. Recalling that the Resting free output concentration is $\xrst[O] = r$, we apply \cref{eq:io-output-mass-action} to bound this difference:
    \begin{equation*}
        a - r \le \nu_{\max} \sum_{Q \in \mathcal{Q}} \xrst[Q] 
        = \nu_{\max} \, r \sum_{Q \in \mathcal{Q}} \xrst[P_Q] \, e^{-\Delta G^{\mathrm{out}}_Q} 
        \le r \left( \nu_{\max} \, Z_{\mathrm{res}}^{\mathrm{out}} \, e^{-\delGoutstg} \right).
    \end{equation*}
    Recognizing that the term in the parentheses is exactly equivalent to $e^{\psi_{\mathrm{out}}}$, we obtain $a - r \le r e^{\psi_{\mathrm{out}}}$. Dividing the entire inequality by $r$ and rearranging for $\gamma := a/r$ completes the proof.
\end{proof}

\Cref{lemma:contrast_upperbound} together with \cref{thm:io-contrast-bound} reveals the following big picture. 
Because the output contrast $\gamma$ acts as the denominator in the overall amplification bound (\cref{thm:io-contrast-bound}), enforcing a high signal contrast inherently restricts the maximum achievable amplification factor $\beta$. Conversely, designing a system with weaker output sequestration (a less negative $\delGoutstg$) mathematically expands the potential for amplification, but leaves the system susceptible to leak and poor contrast. Therefore, to increase the amplification upper bound \emph{without} sacrificing signal clarity, one must instead drive the system with a more favorable input free energy $\delGinpstg$, thereby raising the overall amplification capacity $\phi$.

\begin{remark}\label{cor:sequest_aff}
    \Cref{lemma:contrast_upperbound} implies that 
    to achieve a target output contrast $\gamma > 1$, the output sequestration free-energy must be favorable enough (recall $\delGoutstg < 0$ by convention):
    \begin{equation}\label{corr:affinity_bound}
        \delGoutstg \le \log \left( \frac{\nu_{\max} \, Z_{\mathrm{res}}^{\mathrm{out}}}{\gamma - 1} \right).
    \end{equation}
\end{remark}

In practice, bounding the leakage is critical to avoid prematurely triggering downstream cascades or violating diagnostic sensitivity thresholds. To guarantee a target contrast $\gamma$ for these applications, \cref{corr:affinity_bound} provides a direct design heuristic: the intrinsic sequestration affinity $\delGoutstg$ must be driven sufficiently negative, taking into account the total concentration of available residual partners $Z_{\mathrm{res}}^{\mathrm{out}}$ and their maximum output sequestration size $\nu_{\max}$.

To complete our analysis of the isometric amplifier example, we examine the thermodynamic constraints on the output side of the network. 
In our experiment, the hybridization energy of the output O is $-24~\text{kcal/mol}$ under our conditions. 
While this is substantially less favorable than the hybridization of the input, 
it is in principle compatible with very large contrast depending on $Z_{\mathrm{res}}^{\mathrm{out}}$.
Although $Z_{\mathrm{res}}^{\mathrm{out}}$ is not directly under our control, we could imagine that adding a reasonable excess of strand P results in nanomolar quantities of free P in equilibrium. 
Then the corresponding dimensionless parameters are evaluated as:
\begin{equation*}
    \delGoutstg \approx -44, \quad \nu_{\max} = 2, \quad \text{and} \quad Z_{\mathrm{res}}^{\mathrm{out}} \lesssim 1.8 \times 10^{-11}.
\end{equation*}
Substituting these values into the expression for the output amplification potential yields:
\begin{equation*}
    \psi_{\mathrm{out}} = - \delGoutstg + \log (\nu_{\max} \, Z_{\mathrm{res}}^{\mathrm{out}}) \approx 19.
\end{equation*}
By \cref{lemma:contrast_upperbound}, this potential yields a bound on the output contrast of $\gamma \lesssim e^{19}$.
While operating near this extreme contrast limit would be highly effective for minimizing leak, it heavily penalizes the system's amplification $\beta$ by \cref{thm:io-contrast-bound}. 
Alongside an amplification capacity of $\phi \approx 72$ from the previous section, we see that the maximum achievable amplification factor becomes constrained to:
\begin{equation*}
    \beta \le \frac{\phi}{\log \gamma } \approx \frac{72}{19} \approx 3.8.
\end{equation*}

\input{main-appendix}

\section{Conclusion}
\label{sec:conclusion}

In this work, we first provided a rigorous proof demonstrating that equilibrium dimerization networks are fundamentally incapable of signal amplification, a problem previously addressed only through numerical approximations due to its inherent algebraic complexity. 
Our results help explain the lack of signal amplification in undercomplementary designs proposed by Nikitin~\cite{nikitin2023noncomplementary} as they all rely on a maximum complex size of two.
We subsequently evaluated the amplification potential of equilibrium trimerization networks.
Our investigation of an entropy-driven trimerization network confirmed its capacity to amplify, but revealed a size-shrinking feature that inherently limits its amplification power.
To resolve this, we designed an isometric trimer-based equilibrium amplifier that successfully achieved signal amplification while strictly preserving the length of the output strands, potentially allowing for module composition.
We validated this architecture through the wet-lab implementation of a $2\times$ amplifier, achieving an experimental amplification factor that closely matches this theoretical target.

We finally broadened our scope to investigate fundamental limits on equilibrium amplification in a more general case based on thermodynamic constraints, regardless of the multimericity of the network. 
We demonstrated that the free energy released by the input binding determines the theoretical ceiling of the amplification power; the network design then dictates how much of this capacity is attenuated. 
Specifically, we established that enforcing a high signal contrast fundamentally limits the maximum achievable amplification factor. 
Furthermore, we proved that the theoretical ceiling of this output contrast is itself dictated by the free energy of output sequestration. 
Ultimately, these thermodynamic bounds reveal a universal engineering trade-off: securing a robust, high-contrast signal requires deep output sequestration to suppress leak, which in turn demands a proportionally larger input free-energy budget to successfully drive molecular amplification.

There are two natural directions for future work. First, we aim to experimentally compose multiple amplifier modules to achieve greater than a factor-of-two amplification, using our component-level heuristics (e.g., \cref{rm:maximizing,rm:minimize}) to guide the physical design. Second, we must investigate the tightness of our theorems by testing how closely these physical implementations can approach our established thermodynamic limits for amplification (\cref{thm:io-contrast-bound}) and contrast (\cref{lemma:contrast_upperbound}). 
Ultimately, this leads to a fundamental open question: subject to these thermodynamic limits, does an ``optimal'' equilibrium amplifier or amplifier family exist?

\begingroup
\nolinenumbers

\endgroup

\bibliography{main}



\end{document}

%% file: macro.tex
\DeclareMathAlphabet{\mathpzc}{OT1}{pzc}{m}{it}

\newcommand{\xact}[1][]{%
  \if\relax\detokenize{#1}\relax
    \mathbf{x}%
  \else
    \mathbf{x}_{#1}%
  \fi
}

\newcommand{\xdrn}[1][]{%
  \if\relax\detokenize{#1}\relax
    \mathbf{y}%
  \else
    \mathbf{y}_{#1}%
  \fi
}

\newcommand{\xrst}[1][]{%
  \if\relax\detokenize{#1}\relax
    \mathbf{z}
  \else
    \mathbf{z}_{#1}%
  \fi
}

\newcommand{\xmed}[1][]{%
  \if\relax\detokenize{#1}\relax
    \mathbf{x}_{\text{med}}
  \else
    x_{\text{med},#1}%
  \fi
}

\ifnum\draft=1
\newcommand{\DS}[1]{\textcolor{red}{[David: #1]}}

\newcommand{\HA}[1]{\textcolor{blue}{[Hamid: #1]}}
\newcommand{\todoi}[1]{\todo[inline]{#1}}

\else
\newcommand{\DS}[1]{}

\newcommand{\HA}[1]{}
\renewcommand{\todo}[1]{}
\newcommand{\todoi}[1]{}
\fi

\newcommand{\poly}[1][2]{%
  \if\relax\detokenize{#1}\relax
    \mathcal{P}_
  \else
    \mathbf{y}_{#1}%
  \fi
}

\newcommand{\delGoutstg}{%
    \Delta G_{\mathrm{stg}}^{\mathrm{out}}%
}

\newcommand{\delGinpstg}{%
    \Delta G_{\mathrm{stg}}^{\mathrm{inp}}%
}

\newcommand{\ic}[1][]{%
  \if\relax\detokenize{#1}\relax
    \text{IC}
  \else
    \text{IC}(#1)%
  \fi
}

\newcommand{\rc}[1][]{%
  \if\relax\detokenize{#1}\relax
    \text{RC}
  \else
    \text{RC}(#1)%
  \fi
}

%% file: main-appendix.tex

\begin{toappendix}
    
\section{Materials and Methods}\label{sec:methods}

\subsection{Sequence Design and NUPACK}

The sequence design for the amplifier was computationally optimized to enforce the strict thermodynamic partial order of binding affinities required to drive the cascade: $K_{\text{I:A:X}} > K_{\text{W:X}} > K_{\text{W:P}} > K_{\text{O:O:P}}$. While the physical functionality of the network is ultimately dictated by actual thermodynamic free energies rather than raw mismatch counts, we employed targeted mismatch quotas as a generative heuristic to establish this hierarchy. Specifically, the algorithm assigned 0, 3, 6, and 9 mismatches to the respective complexes, spacing them in increments of three to ensure distinct energetic separation (the specific sequences and spatial positions of these mismatches are detailed in \cref{tab:mismatch_sequences} and illustrated in \cref{fig:mismatch_positions}).

\begin{table}[htbp]
    \centering
    \caption{DNA sequences for the designed isometric trimer-based amplifier.}
    \label{tab:mismatch_sequences}
    \begin{tabular}{cl}
        \toprule
        \textbf{Strand Name} & \textbf{Strand Sequence (5' to 3')} \\
        \midrule
        I & \texttt{CGATGTAGGAAAGGGATTACGTTTGAAGAT} \\
        A & \texttt{CGAAGAATGCGAAGACATGCATTGCAACAT} \\
        X & \texttt{ATCTTCAAACGTAATCCCTTTCCTACATCGATGTTGCAATGCATGTCTTCGCATTCTTCG} \\
        W & \texttt{CGAAGTATGCGAAGACATGCATTGCAACATCGAAGTAGGAAAGGGATTTCGTTTGAAGAT} \\
        P & \texttt{ATCTTCAAACGAAATCCTTTTCTTACTTCGATCTTGCAATGCATTACTTTGCATACTTCG} \\
        O & \texttt{CGAAGTATGTAAAGTAATGCGTTTGAAGAT/3Rox\_N/} \\
        \bottomrule
    \end{tabular}
\end{table}

\begin{figure}[bhtp]
    \centering 
    \includegraphics[width=1\textwidth]{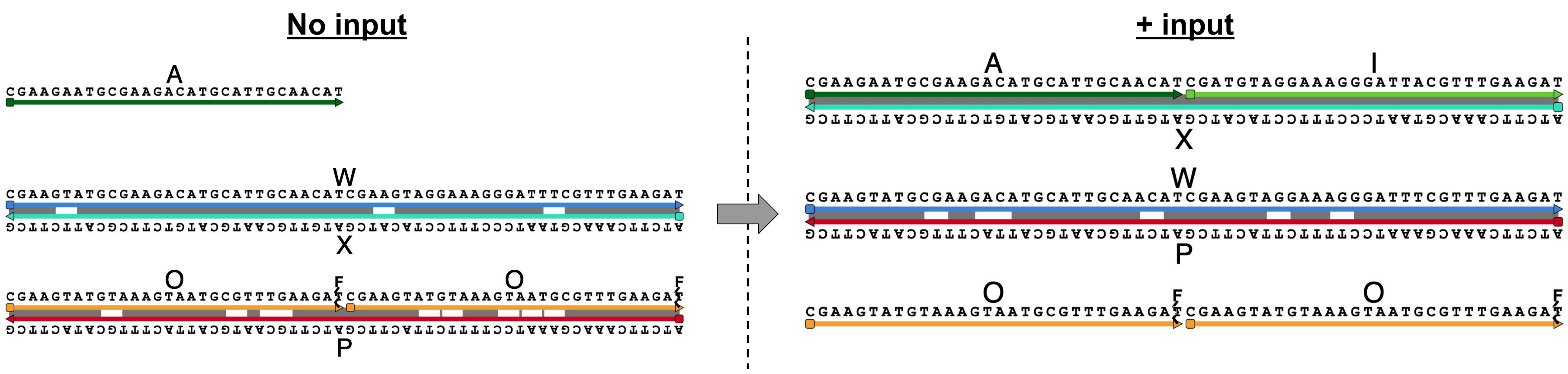} 
    \caption{Sequence-level architecture of the proposed isometric trimerization-based amplifier (illustrated with {\tt scadnano}). An increased density of white gaps within the gray binding domains denotes a higher frequency of engineered nucleotide mismatches used to thermodynamically tune the partial order of binding affinities.} 
    \label{fig:mismatch_positions} 
\end{figure}

Thermodynamic simulations were performed using NUPACK (DNA parameter set, stacking ensemble) at 0.1 M Na$^+$ and 0.0 M Mg$^{++}$. A maximum complex size of three strands was specified to directly reflect the designed trimerization mechanism. To screen for potential higher-order crosstalk, the maximum complex size was subsequently increased to five; this yielded virtually identical results, confirming the absence of significant unintended associations within the pseudoknot-free structural regime modeled by NUPACK.

\subsection{DNA preparation and quantification}

All DNA oligonucleotides were obtained from a commercial supplier and resuspended in TE buffer (pH 8.0). DNA concentrations were determined by measuring absorbance at 260 nm (A260) using a NanoDrop spectrophotometer, and stock solutions were diluted to the desired working concentrations using the same buffer.

For amplifier experiments, the output strand O concentration was fixed at 300 nM, while other circuit strands were maintained at 150 nM unless otherwise specified. The input strand I concentration was varied from 0 to 200 nM to evaluate the input–output response of the amplifier. In addition, calibration lanes containing known concentrations of the output strand were included in the gel analysis to enable estimation of output concentrations in experimental samples.

\subsection{Annealing procedure}

DNA strand mixtures were assembled in annealing buffer and subjected to a slow annealing protocol designed to bring the system close to thermodynamic equilibrium. The annealing buffer consisted of 40 mM Trizma base, 20 mM acetic acid (without EDTA), 100 mM Na$^+$, 5\% formamide, and 5\% DMSO. Samples were first heated to 95°C for 5 minutes to denature all strands, followed by gradual cooling at 0.1°C/s to 4°C. The samples were then held at 4°C overnight to allow equilibration and formation of the intended DNA complexes prior to analysis.

\subsection{Native polyacrylamide gel electrophoresis}

Reaction products were analyzed using native polyacrylamide gel electrophoresis (PAGE). Gels were prepared as 12\% polyacrylamide gels using a mixture containing 40\% acrylamide/bis-acrylamide solution, $1\times$ TBE buffer, and Milli-Q water, with polymerization initiated by 10\% ammonium persulfate (APS) and TEMED.

Electrophoresis was performed in $1\times$ TBE running buffer (89 mM Tris, 89 mM boric acid, 2 mM EDTA). Samples were mixed with Tritrack native gel loading dye prior to loading. Gels were run at 100 V for 60 minutes at 4°C under non-denaturing conditions, allowing DNA complexes to migrate based on size and conformation while preserving strand hybridization interactions.

\subsection{Gel imaging, calibration curve construction, and quantification}

Gel images were captured using a gel documentation system and quantified using GeneTools. Strand O was labeled with a 3$'$ ROX fluorophore using NHS ester chemistry. Fluorescence imaging was performed using a Syngene G:BOX gel documentation system (Syngene G:BOX). ROX fluorescence was detected using excitation/emission wavelengths of 575/602 nm.

\begin{figure}[!t]
    \centering 
    \includegraphics[width=1\textwidth]{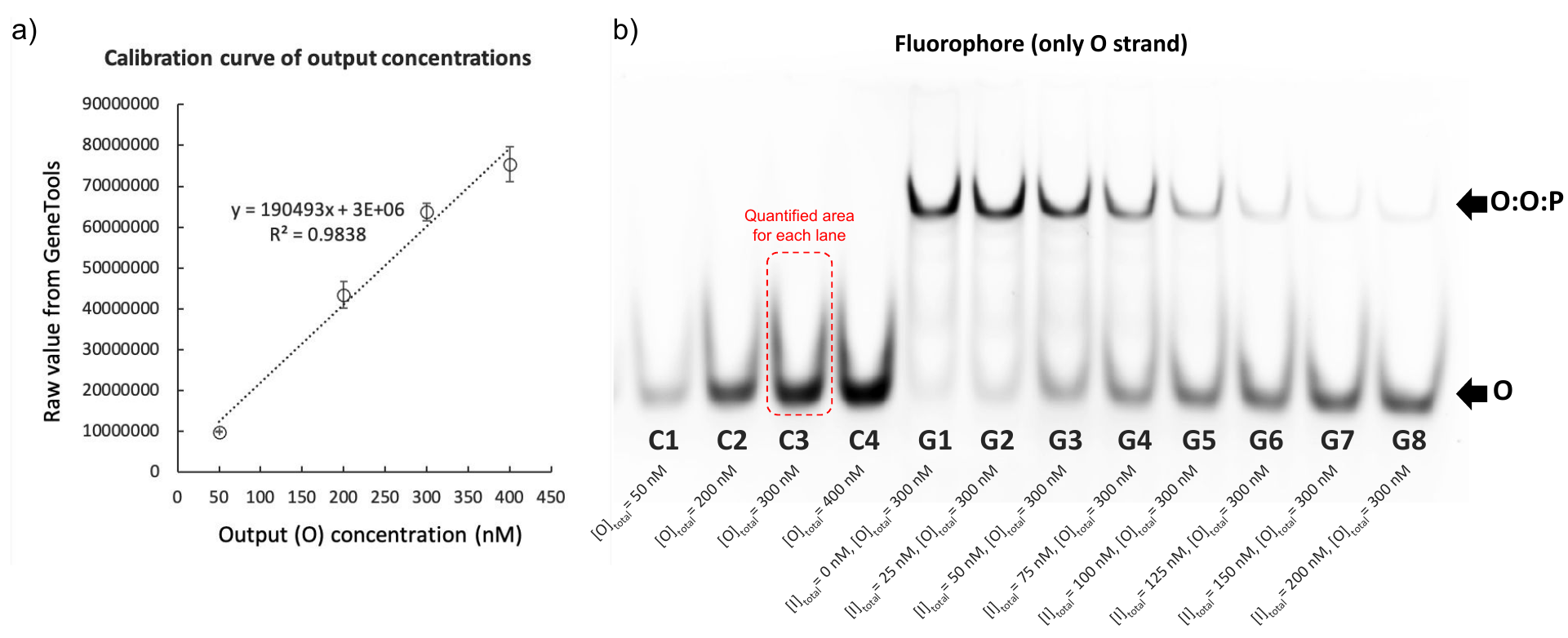} 
    \caption{Quantification of fluorophore-labeled amplifier output using native PAGE. 
    \textbf{(a)} Linear regression of the calibration data. 
    \textbf{(b)} A 12\% native polyacrylamide gel used to quantify the released output. Two distinct bands are observed: the upper band corresponds to the sequestered O:O:P complex, and the lower band corresponds to the free output strand O. The first four lanes (C1--C4) contain known concentrations of the output strand (50, 200, 300, and 400 nM) to establish the calibration curve. Lanes G1--G8 contain amplifier reactions with a fixed total output concentration (300 nM) and increasing input concentrations (0--200 nM). The red dashed box indicates a representative region selected for densitometry using GeneTools; identical region dimensions were applied across all lanes to ensure consistent intensity extraction.} 
    \label{fig:gel} 
\end{figure}

As shown in \cref{fig:gel}b, the native PAGE gel displays the migration of the free output strand O and the O:O:P complex. The first four lanes contain known concentrations of the output strand O (50, 200, 300, and 400 nM) and serve as calibration data. These lanes were used to construct a calibration curve relating band intensity to output concentration. The remaining eight lanes correspond to amplifier reactions in which different concentrations of input strand I were added, resulting in varying levels of O and O:O:P complexes.

For quantification, the band corresponding to the free output strand O in each lane was selected using a fixed rectangular region, and the integrated band area was calculated following local background subtraction. The same region size was applied to all lanes to ensure consistent quantification across samples. For clarity, an example of the quantified region is indicated by the red dashed box in \cref{fig:gel}b, illustrating the region used to extract the GeneTools intensity values for each lane.

To convert GeneTools intensity values into output concentrations, a calibration curve was constructed using the positive control lanes.
The GeneTools integrated band areas from the known output concentrations were plotted against their corresponding concentrations and fitted using a linear regression model to generate the calibration relationship between gel band intensity and output concentration (\cref{fig:gel}a). 
To build this model, the linear regression was performed using all individual replicate data points of the known concentrations rather than their averages. Each experimental condition was measured in three independent replicates ($N = 3$), and the final calculated output concentrations were reported as the mean $\pm$ standard deviation.

\end{toappendix}